\documentclass[letterpaper, 10 pt, conference]{ieeeconf}

\IEEEoverridecommandlockouts  

\usepackage{fix-cm}
\usepackage{etex}

\usepackage{dblfloatfix}

\usepackage{nag}

\makeatletter
\@ifpackageloaded{xcolor}{}{%
\usepackage[table,x11names,dvipsnames,svgnames]{xcolor}%
}
\makeatother

\usepackage{colortbl}

\usepackage{graphicx}
\usepackage{wrapfig}

\definecolorset{RGB}{lyft}{}{Red,194,39,36;Sunset,202,53,33;Orange,205,68,20;Amber,200,117,42;Yellow,242,169,52;Citron,186,188,44;Lime,112,159,33;Green,56,139,31;Mint,45,118,56;Teal,52,133,135;Cyan,60,132,202;Blue,55,94,248;Indigo,64,13,247;Purple,115,42,248;Pink,176,25,145;Rose,176,32,75}

\definecolorset{HTML}{h}{}{grapefruit,ED5565;bittersweet,FC6E51;sunflower,FFCE54;grass,A0D468;mint,48CFAD;aqua,4FC1E9;bluejeans,5D9CEC;lavender,AC92EC;pinkrose,EC87C0;lightgray,F5F7FA;mediumgray,CCD1D9;darkgray,656D78}

\usepackage{cite}

\usepackage{microtype}

\usepackage[american]{babel}

\usepackage{array}
\usepackage{multirow}
\usepackage{booktabs}
\usepackage{makecell} 

\ifcsname labelindent\endcsname
\let\labelindent\relax
\fi
\usepackage[inline]{enumitem}

\usepackage{subfig}

\newenvironment{lenumerate}[2][]
{\begin{enumerate}[label=(#2\arabic*),leftmargin=0.2in,itemindent=0.15in,#1]}
{\end{enumerate}}

\setlist*[enumerate,1]{label={\itshape\arabic*)}}

\makeatletter
\newcommand{\paragraphswithstop}{%
\let\copyparagraph\paragraph%
\renewcommand\paragraph[1]{\copyparagraph{##1.}}%
}
\makeatother

\usepackage[framemethod=tikz]{mdframed}

\makeatletter
\def\namedlabel#1#2{\begingroup
  #2%
  \def\@currentlabel{#2}%
  \phantomsection\label{#1}\endgroup
}
\makeatother

\makeatletter
\def\namedlabelphantom#1#2{\begingroup
  \def\@currentlabel{#2}%
  \phantomsection\label{#1}\endgroup
}
\makeatother

\newcommand{\parunskip}{\bgroup\unskip\parfillskip=0pt \par\egroup}

\input{preamble/math}
\newcommand{\real}[1]{\mathbb{R}^{#1}{}}

\newcommand{\naturals}[1]{\mathbb{N}^{#1}{}}

\DeclarePairedDelimiter{\norm}{\lVert}{\rVert}

\newcommand{\metrica}[3][]{\langle #2, #3\rangle_{#1}}

\newcommand{\de}{\mathrm{d}}

\DeclareMathOperator{\grad}{{grad}}

\DeclareMathOperator{\hess}{{Hess}}

\newcommand{\subjectto}{\textrm{subject to}\;}

\providecommand{\cC}{\mathcal{C}}

\providecommand{\cF}{\mathcal{F}}

\providecommand{\cK}{\mathcal{K}}

\providecommand{\cP}{\mathcal{P}}

\usepackage{units}

  \newcommand{\newcolorlabel}[2]{%
  \expandafter\newcommand\csname #1\endcsname[1]{%
    \tikz[baseline]{\node[text=white,fill=#2,anchor=base,text height=1.3ex,text depth=0.1ex,font=\sffamily\bfseries]{##1}}}%
}

\newcommand{\newcommenter}[2]{%
  \expandafter\newcommand\csname #1\endcsname[1]{%
    \fcolorbox{#2}{#2}{\color{white}\textsf{\textbf{#1}}}
    {\color{#2}##1}}%
  \expandafter\newcommand\csname at#1\endcsname{%
    \fcolorbox{#2}{#2}{\color{white}\textsf{\textbf{@#1}}}
    {\color{#2}}}%
  \expandafter\newcommand\csname #1cite\endcsname[1]{%
    \csname #1\endcsname {[##1]}
  }%
  \expandafter\newcommand\csname #1ref\endcsname[1]{%
    \csname #1\endcsname {$\blacktriangleright$##1}
  }%
  \expandafter\newcommand\csname #1hl\endcsname[2]{%
    \colorbox{#2}{\color{white}\textsf{\textbf{#1}}}\sethlcolor{Azure2}\hl{##2}~%
    \expandafter\ifx\csname commentarrow\endcsname\relax$\leftarrow$\else \commentarrow[#2]\fi~%
    {\color{#2}##1}}%
  \expandafter\newcommand\csname #1st\endcsname[2]{%
    \colorbox{#2}{\color{white}\textsf{\textbf{#1}}}\sout{##2}~%
    \expandafter\ifx\csname commentarrow\endcsname\relax$\leftarrow$\else \commentarrow[#2]\fi~%
    {\color{#2}##1}}%
}
\newcommenter{TODO}{DodgerBlue1}
\newcommenter{rtron}{Green3}

\usepackage{comment}

\usepackage{pdfcomment}

\usepackage{soul}

\usepackage[normalem]{ulem}

\usepackage{csquotes}

\usepackage{suffix}

\usepackage{environ}

\makeatletter
\newsavebox{\boxifnotempty}
\newcommand{\displayifnotempty}[3]{\sbox\boxifnotempty{#2}\setbox0=\hbox{\usebox{\boxifnotempty}\unskip}%
  \ifdim\wd0=0pt
  \else
  #1\usebox{\boxifnotempty}#3%
  \fi%
}

\newcommand{\ifempty}[2]{\setbox0=\hbox{#1\unskip}%
  \ifdim\wd0=0pt%
  #2%
  \fi%
}

\newcommand{\ifnotempty}[2]{\setbox0=\hbox{#1\unskip}%
  \ifdim\wd0>0pt%
  #2%
  \fi%
}
\makeatother

\newcommand{\switchifempty}[3]{\sbox\boxifnotempty{#1}\setbox0=\hbox{\usebox{\boxifnotempty}\unskip}%
  \ifdim\wd0=0pt{}%
  #2%
  \else{}%
  #3%
  \usebox{\boxifnotempty}%
  \fi%
}

\makeatletter
\@ifundefined{chapter}{\usepackage{algorithm}}{\usepackage[chapter]{algorithm}}
\makeatother
\usepackage{algorithmicx}
\usepackage{algpseudocode}
\makeatother%

\usepackage{scrlfile}

\makeatletter
\newcommand*\newstoreddef[1]{
  \BeforeClosingMainAux{%
    \immediate\write\@auxout{%
      \string\restoredef{#1}{\csname #1\endcsname}%
    }%
  }%
}
\newcommand*{\restoredef}[2]{
  \expandafter\gdef\csname stored@#1\endcsname{#2}%
}
\newcommand*{\storeddef}[1]{
  \@ifundefined{stored@#1}{0}{\csname stored@#1\endcsname}%
}
\makeatother

\usepackage{pageslts}
\NewEnviron{tee}{\BODY\typeout{Marker Tee [start] ^^J \BODY ^^JMaker Tee [end]}}

\usepackage{cleveref}
\usepackage{nameref}

\usepackage{tikz}
\usetikzlibrary{calc}
\usetikzlibrary{matrix}
\usetikzlibrary{chains,scopes}
\usetikzlibrary{shapes.geometric}
\usetikzlibrary{arrows.meta}
\usetikzlibrary{decorations.markings}
\usetikzlibrary{decorations.pathreplacing}
\usetikzlibrary{backgrounds}

\tikzset{
  dim above/.style={to path={\pgfextra{
        \pgfinterruptpath
        \draw[>=latex,|->|] let
        \p1=($(\tikztostart)!1.5em!90:(\tikztotarget)$),
        \p2=($(\tikztotarget)!1.5em!-90:(\tikztostart)$)
        in(\p1) -- (\p2) node[pos=.5,sloped,above]{#1};
        \endpgfinterruptpath
      }
    }
  },
  dim double above/.style={to path={\pgfextra{
        \pgfinterruptpath
        \draw[>=latex,|->|] let
        \p1=($(\tikztostart)!3em!90:(\tikztotarget)$),
        \p2=($(\tikztotarget)!3em!-90:(\tikztostart)$)
        in(\p1) -- (\p2) node[pos=.5,sloped,above]{#1};
        \endpgfinterruptpath
      }
    }
  },
  dim below/.style={to path={\pgfextra{
        \pgfinterruptpath
        \draw[>=latex,|->|] let
        \p1=($(\tikztostart)!-1em!-90:(\tikztotarget)$),
        \p2=($(\tikztotarget)!-1em!90:(\tikztostart)$)
        in (\p1) -- (\p2) node[pos=.5,sloped,below]{#1};
        \endpgfinterruptpath
      }
    }
  },
}

\tikzset{
    right angle quadrant/.code={
        \pgfmathsetmacro\quadranta{{1,1,-1,-1}[#1-1]}     
        \pgfmathsetmacro\quadrantb{{1,-1,-1,1}[#1-1]}},
    right angle quadrant=1, 
    right angle length/.code={\def\rightanglelength{#1}},   
    right angle length=2ex, 
    right angle symbol/.style n args={3}{
        insert path={
            let \p0 = ($(#1)!(#3)!(#2)$) in     
                let \p1 = ($(\p0)!\quadranta*\rightanglelength!(#3)$), 
                \p2 = ($(\p0)!\quadrantb*\rightanglelength!(#2)$) in 
                let \p3 = ($(\p1)+(\p2)-(\p0)$) in  
            (\p1) -- (\p3) -- (\p2)
        }
    }
}

\newcommand{\pgfextractangle}[3]{%
    \pgfmathanglebetweenpoints{\pgfpointanchor{#2}{center}}
                              {\pgfpointanchor{#3}{center}}
    \global\let#1\pgfmathresult
}

\usetikzlibrary{shapes.arrows}
\newcommand{\commentarrow}[1][Azure4]{\tikz[baseline=-3pt]{\node[shape border uses incircle, fill=#1,rotate=180,single arrow, inner sep=1pt, minimum size=6pt, single arrow head extend=2pt]{};}}

\tikzset{ax/.style={-latex,line width=2pt}}

\tikzset{camera/.style={fill=Sienna1,fill opacity=0.5},%
image plane/.style={draw=RoyalBlue3,line width=2pt}}

\newcommand{\safeSet}{\cC_h}
\newcommand{\poseSet}{\cP}

\newcommand{\symmetricM}[1]{\mathbb{S}^{#1}_{++}}
\newcommand{\projectionM}[1]{P_{#1}}
\newcommand{\hatMap}[1]{ {\left [ #1 \right]}_{\times} }

\newcommand{\agentCoordinates}{x_p}

\newcommand{\cameraAxis}{e_c}
\newcommand{\bearing}{\beta}

\newcommand{\derr}{\tilde{d}}

\newcommenter{todo}{Firebrick1}
\newcommenter{bt}{RoyalPurple}
\newcommenter{schiano}{Green}

\title{\LARGE \bf
    A Control Barrier Function Candidate for Quadrotors\\with Limited Field of View
}

\author{Biagio Trimarchi$^1$, Fabrizio Schiano$^2$, Roberto Tron$^{3}$
  \thanks{$^1$  Biagio Trimarchi is with the
    Center for Research on Complex Automated Systems (CASY), Department
    of Electrical, Electronic and Information Engineering (DEI), University
    of Bologna, Bologna, Italy
    {\tt\small biagio.trimarchi@unibo.it}%
  }
  \thanks{$^2$Fabrizio Schiano is with Leonardo S.p.a., Leonardo Innovation Labs, Rome, Italy
    {\tt\small fabrizio.schiano@leonardo.com}%
  }
  \thanks{$^{3}$Roberto Tron is with the Department of Mechanical Engineering,
    Boston University, 110 Cummington Mall, MA 02215, United States
    {\tt\small tron@bu.edu}}%
}

\begin{document}

\maketitle
\thispagestyle{empty}
\pagestyle{empty}

\begin{abstract}
  The problem of control based on vision measurements (bearings) has been amply studied in the literature; however, the problem of addressing the limits of the field of view of physical sensors has received relatively less attention (especially for agents with non-trivial dynamics). The technical challenge is that, as in most vision-based control approaches, a standard approach to the problem requires knowing the distance between cameras and observed features in the scene, which is not directly available.
  Instead, we present a solution based on a Control Barrier Function (CBF) approach that uses a \emph{splitting} of the original differential constraint to effectively remove the dependence on the unknown measurement error. Compared to the current literature, our approach gives strong robustness guarantees against bounded distance estimation errors.
  We showcase the proposed solution with the numerical simulations of a double integrator and a quadrotor tracking a trajectory while keeping the corners of a rectangular gate in the camera field of view.
\end{abstract}

\section{Introduction}
Due to payload and battery power limitations, visual sensors (such as stereo and depth cameras) are the most common sensor type for controlling small Unmanned Aerial Vehicles (UAVs) and are extensively used to accomplish various tasks, such as localization, state estimation, and target tracking. However, these sensors come with two penalizing limitations that may hinder the performance of the algorithms that rely on them: their limited field of view and the absence of a fast and reliable way to obtain accurate distance estimation.

Addressing the limited field of view of visual sensors in vision-based control algorithms design has always been an active field of research \cite{Chaumette2006Tutorial, Gans2011Camera}, and it is still active today, especially on quadrotors, for which there are  additional challenges due to under-actuation and the high speed movement. 
The problem is tackled by the authors of \cite{penin2017Vision, Penin2018RH}, which propose a receding horizon control scheme to track both static targets with known position \cite{penin2017Vision} and moving targets of known dimension in the presence of occlusion \cite{Penin2018RH}. In the context of static target tracking, the authors of \cite{Murali2019Perception, Spasojevic2020Perception} proposed two strategies to take the limited field of view into account during trajectory planning and showed how this improved both the localization capability of the UAV and reduced the tracking error of the controller. The issue is also discussed for drones equipped with down-facing cameras in \cite{Roque2020linearMPC}, which proposes a linear MPC scheme around a quasi-hovering state, and in \cite{zheng2016image}, which proposes a Control Barrier Function \cite{ames2019control} approach to enforce the constraint. Another interesting contribution to the field is given by \cite{Lu2022Flight}, in which the field of view constraint is enforced with a non-linear MPC-based control scheme in the context of collision-free navigation.

The works mentioned so far suppose either the positions of the features to track to be known or that is possible to measure or estimate their depth in the image from onboard sensors. More often than not, however, this assumption does not hold. Knowing the exact relative position requires accurate localization algorithms, and the depth estimated by commercially available cameras is often noisy, or reliable only over a limited range \cite{Real435, Zed2}. For these reasons, there has been active research to develop bearing-only algorithms for navigation in both the single- \cite{Tron2014Bearing, Greenawalt2019Bearing} and multi-agent settings \cite{schiano2016rigidity, schiano2017bearing, schiano2018dynamic}. In this line of work, the authors of \cite{4435005} propose a bearing-only navigation strategy inspired by biological agents, and the authors of \cite{Karimian2021Bearing} propose one that is robust to partial occlusion; instead, in \cite{schiano2018dynamic} use only bearing measurements to coordinate a group of UAVs and estimate their relative pose.

To address both of the mentioned issues, we present a control law robust to distance measurement errors, based on the theory of Control Barrier Functions, that mathematically guarantees that a set of visual features remains in the field of view of an agent equipped with a camera. The presented approach is generic, and can be applied to any first- or second-order controlled mechanical system, such as quadrotors. As a byproduct of our analysis, we prove the viability of the approach even when distances are not measured but belong to a known (conservative) range. \Cref{fig:gates} illustrates a possible application of the proposed control strategy in a drone racing setting, where a quadrotor needs to keep the target gates in their field of view as long as possible \cite{kaufmann2023champion}.

\begin{figure}
  \centering
  \usetikzlibrary{math}

\newcommand{\drawXs}[5]{
    \tikzmath{
        \angleRad = #4;
        \xOffset = #3 * cos(\angleRad);
        \yOffset = #3 * sin(\angleRad);
        \xPos1 = #1 - \xOffset;
        \yPos1 = #2 - \yOffset;
        \xPos2 = #1 + \xOffset;
        \yPos2 = #2 + \yOffset;
    }
    \draw[red, thick, rotate around={#4:(\xPos1,\yPos1)}] (\xPos1 - #5, \yPos1 - #5) -- (\xPos1 + #5, \yPos1 + #5);
    \draw[red, thick, rotate around={#4:(\xPos1,\yPos1)}] (\xPos1 - #5, \yPos1 + #5) -- (\xPos1 + #5, \yPos1 - #5);
    \draw[red, thick, rotate around={#4:(\xPos2,\yPos2)}] (\xPos2 - #5, \yPos2 - #5) -- (\xPos2 + #5, \yPos2 + #5);
    \draw[red, thick, rotate around={#4:(\xPos2,\yPos2)}] (\xPos2 - #5, \yPos2 + #5) -- (\xPos2 + #5, \yPos2 - #5);

    \draw[red, dashed] (\xPos1, \yPos1) -- (\xPos2, \yPos2)
}

\scalebox{0.87}{\begin{tikzpicture}

\coordinate (center) at (0, 0);
\def\armLength{0.4};
\def\inclination{-20}
\def\rotorRadius{0.2}

\draw[thick, blue] (0,0)
.. controls (2, 3) and (3, 2) .. (3,1) ..
controls (3, -2) and (5,1) .. (6,1) .. controls (7, 1) and (7, 0) .. (7, -1) .. controls (7, -3) and (5, -4) .. (4, -4) .. controls (3, -4) and (-2, -4) .. cycle;

\draw[thick] ($(center) + (\inclination:-(\armLength - \rotorRadius)$) -- ($(center) + (\inclination:\armLength-\rotorRadius)$);
\draw[thick] ($(center) + (\inclination+90:-(\armLength-\rotorRadius)$) -- ($(center) + (\inclination+90:\armLength-\rotorRadius)$);

\draw[thick] ($(center) + (\inclination:\armLength)$) circle [radius=\rotorRadius];

\draw[thick] ($(center) + (\inclination:-\armLength)$) circle [radius=\rotorRadius];

\draw[thick] ($(center) + (\inclination+90:-\armLength)$) circle [radius=\rotorRadius];

\draw[thick] ($(center) + (\inclination+90:\armLength)$) circle [radius=\rotorRadius];

    \def\coneRadius{3.0}
    \def\coneOrientation{\inclination + 45}
    \def\coneHalfOpening{60}
        
    \coordinate (P) at (0,0);
    \coordinate (Q1) at (0.7,2.5);
    \coordinate (Q2) at (2.5,3);
       
    \fill[blue!30, opacity=0.5] (P) -- ++(\coneOrientation-\coneHalfOpening:\coneRadius) arc (\coneOrientation-\coneHalfOpening:\coneOrientation+\coneHalfOpening:\coneRadius) -- cycle;
        
    \draw[dashed] (P) -- ++(\coneOrientation:\coneRadius);

    \draw[black] (P) + (\coneOrientation - \coneHalfOpening: 1.0) arc (\coneOrientation-\coneHalfOpening:\coneOrientation:1.0);

    \filldraw (P) + (\coneOrientation - \coneHalfOpening/2.5:1.0) node [below, right] {$\psi_F$};
    \filldraw (P) + (\coneOrientation:\coneRadius-0.2) node [above] {$e_c$};

    \drawXs{1.3}{1.5}{0.5}{-45.0}{0.1};
    \drawXs{3.4}{-0.15}{0.5}{30.0}{0.1};
    \drawXs{7.0}{-1.0}{0.5}{0.0}{0.1};
    \drawXs{6.0}{-3.2}{0.5}{105.0}{0.1};
    \drawXs{2.2}{-3.8}{0.5}{90.0}{0.1};
    \drawXs{-0.5}{-2.0}{0.5}{00.0}{0.1};
    
\end{tikzpicture}}
  \caption{A 2-D visual task problem: a quadrotor traversing a race circuit needs to keep the gates (depicted in red) as long as possible in its field of view (the blue cone) to orient itself and race through the circuit.}
  \label{fig:gates}
\end{figure}
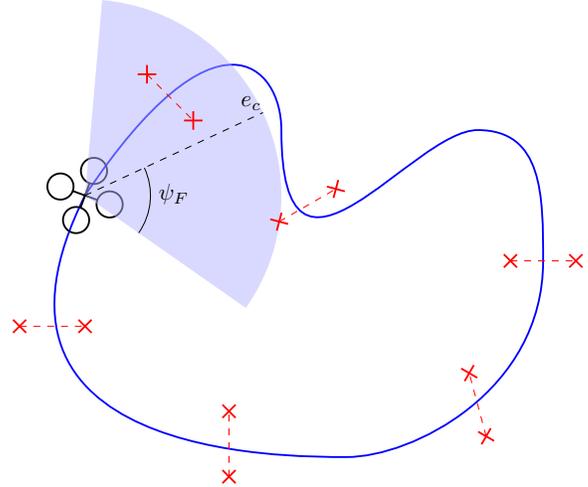


\section{Notation and Preliminaries}
\label{sec:Preliminaries}
In this section, we introduce the notation and the theoretical fundamentals that will be used throughout the paper. We refer the reader to the cited literature for proofs and derivations. 

\subsection{Notation}
In the following $\real{}$ denotes the set of reals number, $SO(3)$ denotes the special orthogonal group of dimension 3, and $\symmetricM{n}$ denotes the set of positive definite and symmetric matrices of dimension $n \times n$. Given two vectors $v_1, v_2 \in \real{2}$, $\angle(v_1, v_2)$ denotes the angle between the two vectors.

The \emph{hat map} $\hatMap{\cdot}: \real{3} \to \real{3 \times 3}$ associates a vector $v \in \real{3}$ to a skew-symmetric matrix according to the formula:
\begin{align*}
  \hatMap{v}
  =
  \begin{bmatrix}
    0 & -v_3 & v_2 \\
    v_3 & 0 & -v_1 \\
    -v_2 & v_1 & 0
  \end{bmatrix}
\end{align*}

We say that a function $\gamma: \real{} \to \real{}$ is a \emph{class $\mathcal{K}$ function} if $\gamma(0) = 0$ and it is strictly increasing. We denote as $\gamma'(x) = \frac{\de \gamma}{\de x}(x)$ the first derivative of $\gamma$ with respect to its argument.

In the following, let $\poseSet = \real{3} \times SO(3)$, and let $T_x \poseSet$ denotes its tangent space at $x \in \poseSet$. $\poseSet$ is a Riemannian manifold\cite{boumal2023introduction}, equipped at each point $x \in \cP$ with the usual metric (inner product) ${\metrica{v}{w}}_{x} : T_x \poseSet \times T_x \poseSet \to \real{}$. By identifying the tangent space of $SO(3)$ with $\real{3}$, we can consider $T_x \poseSet = \real{6}$ and define the inner product as
\begin{align*}
  {\metrica{v}{w}}_x = \sum_{i=1}^{6} v_i w_i.
\end{align*}

Given a twice continuously differentiable function $h : \poseSet \to \real{}$  the function $\grad h : \poseSet \to T_x\poseSet$ will denote its \emph{gradient} and the function $\hess h [\cdot]: T_x\poseSet \to T_x\poseSet$ will denote its \emph{Hessian}.

Since $\cP$ is a product manifold, each point $x \in \poseSet$ is a tuple $(p, R)$, and we can identify the components of $\grad$ and $\hess$ corresponding to each component.
We define the two linear maps $\grad_p h : T_p \real{3} \to \real{3}$ and $\grad_R h : T_R SO(3) \to \real{}$ such that, for any tangent vectors $v \in T_p \real{3}$ and $\omega \in T_R SO(3)$ it holds that:
\begin{align*}
  \metrica{\grad h}{ \begin{bmatrix}
    v \\ 0
  \end{bmatrix}} = \metrica{\grad_p h}{v}, \
  \metrica{\grad h}{ \begin{bmatrix}
    0 \\ \omega
  \end{bmatrix}} = \metrica{\grad_R h}{\omega}.
\end{align*}

Respectively, for the  Hessian, we define two linear maps $\hess_p h [\cdot]: T_p\real{3} \to T_p\real{3}$ and $\hess_R h[\cdot]: T_R SO(3) \to T_R SO(3)$ such that:
\begin{align*}
  \hess h \left[
  \begin{bmatrix}
    v \\ 0
  \end{bmatrix} \right] = \hess_p [v], \
  \hess h \left[ \begin{bmatrix}
    0 \\ \omega
  \end{bmatrix} \right] = \hess_R [\omega].
\end{align*}

Given a vector field $f: \poseSet \to \poseSet$ and a function $h: \poseSet \to \real{}$, the \emph{Lie Derivative} of $h$ with respect to $f$ at $x \in \poseSet$ is denoted as $L_f h(x)$ and is defined as $L_f h(x) = \metrica{\grad_p h}{f(x)}$.

\subsection{Preliminaries}
In this section, we review some simple facts.
\begin{proposition}
  \label{th:var_alpha}
  Consider $x \in \real{}$ and the first order differential equation
  \begin{align}
    \label{eq:var_alpha}
    \dot{x} = -\alpha(x, t)
  \end{align}

  where $\alpha: \real{} \times \real{} \to \real{}$ is continuously differentiable, locally Lipschitz in both arguments and, for each $t \in \real{}$, $\alpha(t, x)$ is a class $\cK$ function of $x$. Then $x = 0$ is a globally uniformly asymptotically stable equilibrium point for (\ref{eq:var_alpha}). Moreover, if $x(0)\geq 0$ then $x(t)\geq 0$ for all $t>0$.
\end{proposition}
\begin{proof}
  Consider the Lyapunov function $V = \frac{1}{2} x^2$; the time derivative is
  \begin{align*}
    \dot{V} = -\alpha(t, x) x.
  \end{align*}
  Since $\alpha(t, x)$ is a class $\cK$ function for any $t$, we have that $\dot{V} = 0$ if $x = 0$, and $\dot{V} < 0$ otherwise. The first part of the claim then follows directly from the Lyapunov theorem for nonlinear autonomous systems \cite{khalilnonlinear}. For the second part, assume by way of contradiction that there exists $t'>0$ such that $x(t')<0$; then, from the continuity of the solution to the differential equation, there exists a $t''$ such that $x(t'')=0$; however, this implies $\dot{x}(t'')=0$ (for any value of $\alpha$), and $x(t)=0$ for all $t>t''$, leading to a contradiction at $t'$.
\end{proof}

\begin{definition}
  Given a unit vector $v \in \real{n}$, the matrix
  \begin{align*}
    P_v = I - v v^T \in \real{n \times n}
  \end{align*}
  is called the \emph{projection matrix} associated to $v$.
\end{definition}

\begin{lemma}
  Each unit vector $v \in \real{n}$ belongs to the kernel of its projection matrix $P_v$, i.e. $v \in Ker(P_v)$.
\end{lemma}
\begin{proof}{}
  The statement follows from the definition:
  \begin{align*}
    P_v v = (I - vv^T) v = v - vv^Tv = v - v = 0.
  \end{align*}
\end{proof}

\subsection{Control Barrier Functions}
Consider an affine dynamical system of the form:
\begin{align}
  \label{eq:affine_dynamics}
  \dot{x} = f(x) + g(x) u,
\end{align}
where $x \in \mathcal{X} \subset \real{n}$ is the \emph{state} of the system, $u \in \real{m}$ is the \emph{control input}, and $f: \mathcal{X} \to \real{n}$ and $g: \mathcal{X} \to \real{n \times m}$ are continuously Lispchitz functions.
\begin{definition}
  Given a continuously differentiable function $h: \real{n} \to \real{}$, we define the \emph{safe set} $\cC_h \subset \real{n}$ associated to the function $h$ as
  \begin{align}
    \label{eq:safe_set}
    \cC_h = \left \{ x \in \real{n} \mid h(x) \geq 0 \right \}, \nonumber\\
    \partial \cC_h = \left \{ x \in \real{n} \mid h(x) = 0 \right \}, \\
    Int(\cC_h) = \left \{ x \in \real{n} \mid h(x) > 0 \right \}. \nonumber
  \end{align}
  .
\end{definition}

\begin{definition}
  A continuously differentiable function $h : \real{n} \to \real{}$ is a \emph{control barrier function} for the dynamical system \eqref{eq:affine_dynamics} if there exists a class $\mathcal{K}$ function $\gamma$ such that:
  \begin{align*}
    \max_{u \in \real{m}} L_f h(x) + L_g h(x) u + \alpha h(x) \geq 0, \forall x \in Int(\safeSet)
  \end{align*}
\end{definition}

\begin{definition}
  We say that the \emph{relative degree} of $h: \real{n} \to \real{}$ with respect to $u$ in equation \eqref{eq:affine_dynamics} is $r \in \naturals{}$, if:
  \begin{align*}
    L_g L_f^k h(x) &= 0, \ \forall i = \{0, 1, \ldots, r-2\} \\
    L_g L_f^{r-1} h(x) &\neq 0.
  \end{align*}
\end{definition}

\begin{definition}[HOCBF \cite{xiao2022HighOrder}]
  \label{def:HOCBF}
  For a $r^{th}-$order differentiable function $h: \real{n} \to \real{}$, we consider a sequence of functions $ \psi_i : \real{n} \to \real{} $ defined as:
  \begin{align}
    \label{eq:hocbf}
    \psi_i(x) = \dot{\psi}_{i-1}(x) + \gamma_i(\psi_{i-1}(x)), \quad i \in \{ 1, \ldots, r \}
  \end{align}
  where $\gamma_i : \real{} \to \real{}$ are class $\mathcal{K}$ functions and $\psi_0(x)= h(x)$.

  We further define a sequence of sets $C_i$ associated with \eqref{eq:hocbf} as:
  \begin{align*}
    C_i = \{ x \in \real{n} \mid \psi_{i-1}(x) \geq 0 \}, \quad i \in \{ 1, \ldots, r \}.
  \end{align*}

  The function $h(x)$ is said to be a candidate \emph{High Order Barrier Function} (HOCBF) of relative degree $r$ for system \eqref{eq:affine_dynamics} if there exist differentiables class $\mathcal{K}$ functions $\gamma_i$, $i \in \{ 1, \ldots, r\}$ such that
  \begin{align*}
    \sup_{u \in \real{m}} [
    & L_f^r h(x) + L_g L_f^{r-1} h(x) u + \\
    & O(h(x)) + \gamma_r(\psi_{r-1}(h(x))
      ] \geq 0
  \end{align*}
  for all $x \in \bigcup_{i=1}^r C_r$, where we have defined
  \begin{align*}
    O(h(x)) = \sum_{i=1}^{r} L_f^i (\gamma_{r-i}(\psi_{r-i-1}(x)))
  \end{align*}

  Given a HOCBF $h(x)$, we define the following set of control inputs:
  \begin{align}
    \label{eq:safe_controller}
    K_{hocbf}(x) = \{
    u \in \real{m} \mid & L_f^r h(x) + L_g L_f^{r-1} h(x) u + \\
                        & O(h(x)) + \gamma_r(\psi_{r-1}(x)) \geq 0
                          \} \nonumber
  \end{align}
\end{definition}

\begin{theorem}[\cite{xiao2022HighOrder}]
  \label{th:HOCBF}
  Given an HOCBF $h : \real{n} \to \real{}$ and its associated sets $C_i$, $i \in \{1, \ldots, r \}$, if $x(t_0) \in \bigcup_{i=1}^r C_r$, then any Lipschitz continuous controller $u(x) \in K_{hocbf}(x)$ renders the set $\bigcup_{i=1}^r C_r$ forward invariant.
\end{theorem}

\section{Control Barrier Function Constraint for Bearing Measurements}
\label{sec:Approach}

In this section, we show how to formulate a field of view constraint as a Control Barrier Function, and the necessary manipulations to impose the constraint when only unreliable distance measurements are available.

\paragraph{Agent kinematics}
We will consider an autonomous agent moving in space modeled as a rigid body whose pose with respect to an inertial reference frame is denoted as $\agentCoordinates = (p, R) \in \real{3} \times SO(3)$, where $p \in \real{3}$ denote the position of the rigid body and $R \in SO(3)$ its rotation transforming coordinates from the body to the inertial frame.

\paragraph{Sensing model}
The agent is equipped with a visual sensor whose field of view is modeled as an infinite half-cone centered in $p$ and with total aperture $2\psi_{F}$, with $\psi_F \in \left (0, \frac{\pi}{2} \right)$ (see Fig. \ref{fig:gates}~for reference). We denote with $\cameraAxis \in \real{3}$ the orientation of the optical axis of the visual sensor in the body frame. Formally, we denote the field of view of the agent $\cF(x) \subset \real{3}$ as:
\begin{align*}
  \cF(x) = \left \{ q \in \real{2} \mid \angle(\bearing(p), R\cameraAxis)  < \psi_{F}  \right \}.
\end{align*}

\paragraph{Control task}
We suppose that a reference controller $u^*(x) : \mathcal{X} \to \real{m}$ is given, able for example to track trajectories with the use of visual features. We let ${\{ q_i \}}_{i=0}^N \subset \real{3}$ denote a set of $N$ scene features the agent needs to keep in sight during its operation, that may be useful either for localization or navigation. We use $d_i (p) = \norm{q_i - p}$ to denote the distance of point $i$ from the agent position $p$, and $\bearing_i (p) = \frac{q_i - p}{\norm{q_i - p}}$ for its associated bearing.
Let $\hat{d}_i(p, t) : \real{} \times \real{} \to \real{}$ be an estimate of the quantity $d_i(p)$ available to the agent; we then define the relative error as $\tilde{d} (p, t) = \frac{d (p)}{\hat{d} (p, t)}$.

\begin{problem}\label{problem:fov invariance}
  The goal of the agent is to compute a control input $u$ that is as close as possible to $u^*$ while keeping each point of interest inside its field of view, i.e., at every time instant $t$,
  \begin{equation}\label{eq:goal}
    q_i \in \cF(x), \quad \forall i \in \{ 1, \ldots, N\}.
  \end{equation}
\end{problem}
\begin{remark}
  In the following discussion, we will prove that the proposed approach is robust to bounded distance readings error, that is, if the multiplicative $\derr(t)$ is contained in an interval (which may be estimated from the sensor data sheet or by computing some bounds from the algorithm used to calculate the distance). 
\end{remark}

\paragraph{Candidate CBF function}
Condition \eqref{eq:goal} for each feature point $q_i$ can be equivalently characterized by a function $h_i : \real{3} \times SO(3) \to \real{}$ defined as
\begin{align}
  \label{eq:cbf}
  h_i(x) = \bearing_i^T(p) R \cameraAxis - \cos{\psi_F};
\end{align}
It can be verified that:
\begin{itemize}
\item $h_i(x) > 0$, when $q_i \in int \cF$
\item $h_i(x) = 0$ when $q_i \in \partial \cF$
\item $h_i(x) < 0$, when $q_i \not\in int \cF$
\item $h_i(x)$ is smooth anywhere except when $p = q_i$, where it is not defined.
\end{itemize}
and, as such, is a candidate Control Barrier Function. We can then define the safe sets $\cC_{h_i} \subset \real{3} \times SO(3)$ according to \eqref{eq:safe_set}.

For the sake of exposition, in the following analysis, we will consider $N=1$ and drop the subscript accordingly. At the end of this section, we will explain how to handle the general case $N>1$. Furthermore, we will drop the dependence on the agent and feature positions when referring to the bearing $\beta$ and distance function $d$ associated with the feature.

\subsection{Velocity Control}
We first analyze the case in which we have control authority over the linear and angular velocities of the rigid body. By denoting the linear velocity of the body, expressed in the inertial frame, by $v\in T_p \real{3}$ and its angular velocity, expressed in the body frame, by $\omega\in T_R SO(3)$, the equations of motion become
\begin{align*}
  \dot{p} &= v \\
  \dot{R} &= R \hatMap{\omega}
\end{align*}
Comparing with \eqref{eq:affine_dynamics}, we then have $x = \agentCoordinates$ and $u = ( v, \omega )$.

As is usual in the CBF literature, we aim to solve Problem~\ref{problem:fov invariance} by solving the following optimization problem at each time instant:
\begin{subequations}
  \begin{align}
    \min_{u \in \real{3}} & \quad \norm{u - u^*}^2 \\
    s.t. & \quad \metrica{\grad h}{u} + \gamma(h(x)) \geq 0. \label{eq:qp_single_classic}
  \end{align}
\end{subequations}
for a given class $\cK $ function $\gamma : \real{} \to \real{}$.
The constraint in \eqref{eq:qp_single_classic} expands to
\begin{align*}
  \metrica{\grad_p h}{v}
  + \metrica{\grad_R h}{\omega}
  + \gamma (h(x)) &\geq 0,
\end{align*}
which can be further expanded (see the Appendix) as
\begin{align*}
  -\frac{1}{d}\cameraAxis^T R^T \projectionM{\bearing}^T v - \bearing^T R  \hatMap{\cameraAxis} \omega + \gamma (h(x)) &\geq 0.
\end{align*}

Notice that the term that multiplies $v$ depends explicitly on the multiplicative factor $\frac{1}{d}$, which, in our assumptions, is only approximately known through the estimate $\hat{d}$. To get around the issue we propose to substitute the previous constraint with the following triplet of constraints:
\begin{equation}
  \begin{aligned}
    \label{eq:splitted1}
    -\frac{1}{\hat{d}}
    \cameraAxis^T R^T \projectionM{\bearing}^T v + c_1 \gamma(h(x))&\geq 0, \\
    - \bearing^T R \hatMap{\cameraAxis} \omega + c_2 \gamma (h(x)) &\geq 0, \\
    c_1 + c_2 = \gamma_0 &> 0.
  \end{aligned}
\end{equation}
The following theorem shows that, under appropriate assumptions, there exist coefficients $c_1,c_2\in\real{}$ such that enforcing \eqref{eq:splitted1} implies satisfaction of \eqref{eq:qp_single_classic}.

\begin{theorem}
  \label{th:split_first_order}
  Let $\gamma : \real{} \to \real{}$ be a class $\cK$ function. Given a nominal control law $u^*(x) : \real{3} \times SO(3) \to \real{3}$, and constants $\gamma_0, d_m, d_M > 0$, such that $d_m < 1 < d_M$, the control law $u$ resulting from the solution of the following quadratic program (which is assumed to be feasible):
  \begin{subequations}
    \begin{align}
      \min_{\substack{u \in \real{3} \\ c_1, c_2 \in \real{}}}
      & \quad \norm{u - u^*(x)}^2 \\
      \label{cstr: velocity}
      \subjectto          & \quad - \frac{1}{\hat{d}}\cameraAxis^T R^T \projectionM{\bearing}^T v + c_1 \gamma(h(x))\geq 0 \\
      \label{cstr: angular}
      & \quad - \bearing^T R \hatMap{\cameraAxis} \omega + c_2 \gamma(h(x)) \geq 0 \\
      \label{cstr:scale}
      & \quad c_1 + c_2 = \gamma_0  \\
      \label{cstr:positive}
      & \quad \frac{\gamma_0}{1 - d_{M}} < c_2 < \frac{\gamma_0}{1 - d_{m}}
    \end{align}
  \end{subequations}

  renders the set $\safeSet$ forward invariant whenever $\derr \in [d_m, d_M]$.
\end{theorem}

\begin{remark}
  At a high level, the idea of the proof is to show that imposing \eqref{cstr: velocity} and \eqref{cstr: angular} is equivalent to imposing the CBF constraint with a time-varying class $\cK$ function $\bar{\gamma}$ which is different from $\gamma$. The constraint \eqref{cstr:positive} is then necessary to ensure that the CBF constraint from $\bar{\gamma}$ is valid. Finally, constraint \eqref{cstr:scale} can be interpreted as fixing the scale for $\bar{\gamma}$.
\end{remark}

\begin{proof}{}
  We multiply \eqref{cstr: velocity} by $\frac{1}{\derr}$ and obtain, since $\derr > 0$,
  \begin{align*}
    -\frac{1}{d}\cameraAxis^T R^T \projectionM{\bearing}^T v + \frac{c_1}{\derr} \gamma (h(x))\geq 0;
  \end{align*}
  summing with \eqref{cstr: angular} we obtain
  \begin{align*}
    \metrica{\grad h}{u} + \left ( c_2 + \frac{c_1}{\derr}  \right) \gamma (h(x)) \geq 0.
  \end{align*}

  To ensure forward invariance, we need $c_2 + \frac{c_1}{\derr} > 0$. We can rewrite this condition as a function of $c_2$ alone by using constraint \eqref{cstr:positive}:
  \begin{align*}
    c_2 + \frac{c_1}{\derr} = c_2 + \frac{\gamma_0 - c_2}{\derr} = \frac{c_2(1-\derr) + \gamma_0}{\derr} > 0.
  \end{align*}
  For the above to hold, we need to study how the function changes for various values of $\derr$. One can easily verify that the above condition can be rewritten as:
  \begin{align*}
    \begin{cases}
      c_2 \in \real{} &\textrm{if } \derr = 1, \\
      c_2 < \frac{\gamma_0}{1 - \derr} &\textrm{if }  0 < \derr < 1, \\
      c_2 > \frac{\gamma_0}{1 - \derr}& \textrm{if } \derr > 1,
    \end{cases}
  \end{align*}
  which is equivalent to constraint \eqref{cstr:positive}.

  It follows that, if $\derr \in [d_m, d_M]$, then $\bar{\gamma} (h, t) = \left ( c_2(t) + \frac{c_1(t)}{\derr(t)} \right ) \gamma(h)$ is a valid CBF function; from \Cref{th:var_alpha} we then have that, if $h(0) > 0$, the solution of the differential equation $\dot{h} = -\bar{\gamma}(h, t) h$ is strictly positive. From the comparison lemma \cite{khalilnonlinear}, we then conclude that the set $\safeSet$ is forward invariant.
\end{proof}

\subsection{Acceleration Control}
Let us now analyze the case in which the agent dynamics is described by a second-order system. Denoting the linear acceleration of the body in the inertial frame by $a \in T_p \real{3}$, and its angular acceleration in the body frame by $\alpha \in T_R SO(3)$, the equations of motion become
\begin{align*}
  \dot{p} &= v \\
  \dot{v} &= a \\
  \dot{R} &= R \hatMap{\omega} \\
  \dot{\omega} &= \alpha
\end{align*}

Comparing with \eqref{eq:affine_dynamics}, we then have $x = (p, R, v, \omega)$ and $u = (a, \alpha)$.
According to \eqref{eq:safe_controller}, for given class $\cK$ functions $\gamma_1, \gamma_2 : \real{} \to \real{}$, the constraint to satisfy in this case is
{\footnotesize\begin{equation*}
  L_f^2 h(x) + L_g L_f h(x) u +
  \gamma_1'(h) L_f h(x) + \gamma_2 (L_f h + \gamma_1(h(x))) \geq 0
\end{equation*}}
which expands to
\begin{align}
  \label{eq:second_order_contraint}
  \metrica{\grad_p h}{a} +
  \metrica{\grad_R h}{\alpha} &+ \nonumber \\
  \metrica{\hess_{p} h [v]}{v} +
  \metrica{\hess_{R} h) [\omega]}{\omega } &+ \nonumber \\
  \metrica{\hess_{p} h [v]}{\omega} +
  \metrica{\hess_{R} h [\omega]}{v} &+ \\
  \gamma_1'(h) \left( \metrica{\grad_p h}{v} +
  \metrica{\grad_R h}{\omega} \right) + \nonumber \\
  \gamma_2 (\metrica{\grad_p h}{v} +
  \metrica{\grad_R h}{\omega} + \gamma_1(h(x))) &\geq 0 \nonumber
\end{align}

Similarly to \Cref{th:split_first_order}, we split  \eqref{eq:second_order_contraint} to obtain two constraints, one that collects all the terms that contain $d$ in the denominator, and one that collects all the terms that do not contain $d$. However, the term $\metrica{\hess_p h [v]}{v}$ contains $d^2$ in the denominator, as it expands (see the Appendix) to
\begin{align}
  \label{eq:hessian_velocity}
  v^T\frac{2(\cameraAxis^T R^T\bearing) I - 3 (\cameraAxis^T R^T\bearing) P(\bearing)  - R \cameraAxis \bearing^T - \bearing \cameraAxis^T R^T}{d^2}  v.
\end{align}
As a consequence, this term cannot belong to either of the two constraints. To deal with this problem, we decompose the velocity $v$ along a component parallel to the bearing, $v_\parallel = \metrica{\bearing}{v}\bearing$, and one perpendicular to it, $v_\perp = \projectionM{\bearing} v$, so that $v = v_\parallel + v_\perp$. Moreover, since $\dot{\bearing} = -\frac{1}{d} \projectionM{\bearing}v$, we have
\begin{align}\label{eq:vperp}
  v_{\perp} = -d \dot{\bearing}.
\end{align}

Now, due to the linearity of the Hessian and the metric, we can rewrite \eqref{eq:hessian_velocity} as follows:\begin{multline}
  \label{eq:hessian_decomposed}
  \metrica{\hess_p h [v]}{v} = \metrica{\hess_p h [v_\perp + v_\parallel]}{v_\perp + v_\parallel} = \\
  = \metrica{\hess_p h [v_\perp]}{v_\perp} + \metrica{\hess_p h [v_\perp]}{v_\parallel} \\
  + \metrica{\hess_p h [v_\parallel]}{v_\perp} + \metrica{\hess_p h [v_\parallel]}{v_\parallel}.
\end{multline}
Expanding the computations, we can see that
$\metrica{\hess_p h [v_\parallel]}{v_\parallel} = 0 \nonumber$, and we can use \eqref{eq:vperp} to simplify a distance factor from the remaining terms containing $d^2$ at the denominator.

As a consequence, under the assumption that $\dot{\bearing}$ is measurable (e.g., through numerical differentiation or the use of an observer), we can compute the various terms in \eqref{eq:hessian_decomposed} without requiring the square of the distance. This leads us to the following claim.

\begin{theorem}
  \label{th:split_second_order}
  Let $\gamma_1, \gamma_2 : \real{} \to \real{}$ be class $\cK$ functions. Without loss of generality, assume $\gamma_1'(0) \geq \gamma_2'(0)$ (otherwise swap the role of the two functions). Define $K = \frac{4\gamma_1'(0) \gamma_2'(0)}{{(\gamma_1'(0) + \gamma_2'(0))}^2}  $. Given a nominal control law $u^*(x) : \real{3} \times SO(3) \to \real{3}$ and constants $\gamma_0, d_m, d_M > 0$, such that $\gamma_0 > K$ and $d_m < 1 < d_M$, then there exists $M > 0$ such that the control law resulting from the solution of the following quadratic program
  \begin{subequations}
    \label{eq:second_op}
    \begin{align}
      \min_{\substack{u \in \real{3} \\ c_1, c_2 \in \real{}}}
      & \quad \norm{u - u^*(x)}^2 \\
      \textrm{s.t. }               & \quad \derr \left(\metrica{grad_p h}{a} +  2 \metrica{\hess_p [v_\perp]}{v_\parallel} \right) + \\
      & \qquad c_1 \left(\gamma'_1 L_f h(x) +\gamma_2 (L_f h(x) + \gamma_1 (h(x)) \right) \geq 0 \nonumber\\
      & \quad \metrica{grad_R h}{\alpha} + \metrica{\hess_R h [\omega]}{\omega} + \\
      & \qquad 2 \metrica{\hess_R h [\omega]}{v}  + \metrica{\hess_p h [v_{\perp}]}{v_\perp} + \nonumber\\
      & \qquad c_2 \left(\gamma'_1 L_f h(x) +\gamma_2 (L_f h(x) + \gamma_1 (h(x)) \right) \geq M \nonumber\\
      & \quad c_1 + c_2 = \gamma_0 \label{eq:scale second}\\
      & \quad \frac{Kd_M - \gamma_0}{d_M - 1} \leq c_2 \leq \frac{Kd_m - \gamma_0}{d_m - 1}
    \end{align}
  \end{subequations}
  renders the set $\safeSet$ forward invariant whenever $\derr \in [d_m, d_M]$.
\end{theorem}
  The terms in \eqref{eq:second_op} are given by (see the Appendix for the full derivation):
  \begin{align*}
    \metrica{\hess_p h [v_\perp]}{v_\perp} &= -\left(\bearing^T R \cameraAxis \right)\left (\dot{\bearing}^T \dot{\bearing} \right) \nonumber \\
    \metrica{\hess_p h [v_\perp]}{v_\parallel} &= \metrica{\hess_p h [v_\parallel]}{v_\perp} =
                                                 \frac{(v^T \bearing) \dot{\bearing}^T P(\bearing) R \cameraAxis}{d} \\
    \grad_p h(x) &= - \frac{\projectionM{\bearing} R \cameraAxis}{d} \\
    \grad_R h(x) &= -\bearing^T R \hatMap{\cameraAxis} \\
    \metrica{\hess_R h[\omega]}{v} &= \metrica{\hess_p h[v]}{\omega}= v^T \frac{\projectionM{\bearing} R \hatMap{\cameraAxis}}{d} \omega \\
    \metrica{\hess_R h [\omega]}{\omega}&= \omega^T \hatMap{R^T \bearing} \hatMap{\cameraAxis} \omega.
  \end{align*}
  \begin{remark}
      At a high level, the proof follows the same reasoning of Theorem \ref{th:split_first_order}. However, it also needs to take into account the time-varying nature of $\gamma_1$ when defining the HOCBF constraint $\psi_2$; this, in turn, depends on $\dot{\tilde{d}}$, which needs to be bounded by $M$.
    \end{remark}

    Before considering the proof of \Cref{th:split_second_order}, we prove the following intermediate results:
    \begin{lemma}\label{lemma:alpha12}
  Let $L = \left (c_2 + \frac{c_1}{\derr} \right )$. There exists two functions $\alpha_1, \alpha_2$ such that
  \begin{equation}\label{eq:alphas equality}
    \alpha_1'(h) \dot{h} + \alpha_2( \dot{h} + \alpha_1(h)) =
    L \left (\gamma_1'(h) \dot{h} + \gamma_2( \dot{h} + \gamma_1(h))\right)
  \end{equation}
    \end{lemma}
    \begin{proof}{}
  Since \eqref{eq:alphas equality} must hold for each $h$ and $\dot{h}$, it must be true also for the special cases where $h = 0$ or $\dot{h}=0$. Remembering that, from the properties of class-$\cK$ functions, $\gamma_1(0)=\alpha_1(0)=0$, eq. \eqref{eq:alphas equality} implies
  \begin{align}
      h = 0 &\Rightarrow L \left( \gamma_1'(0) \dot{h} + \gamma_2(\dot{h}) \right) = \alpha_1'(0) \dot{h} + \alpha_2(\dot{h}),
      \\
      \dot{h} = 0 &\Rightarrow L \gamma_2(\gamma_1(h)) = \alpha_2(\alpha_1(h)).
  \end{align}

  From the above, the following must hold for any $z \in \real{}$:
  \begin{subequations}\label{eq:alphas}
    \begin{align}
      \alpha_1(z) &= \alpha_2^{-1}(\gamma_2(\gamma_1(z))) \\
      \alpha_2(z) &= L \left( \frac{\de \gamma_1}{\de h}(0) z + \gamma_2(z) \right) - \alpha_1'(0) z
    \end{align}
  \end{subequations}
  This provides an expression of $\alpha_1$ as a function of $\alpha_2$, and of $\alpha_2$ as a function of the scalar $\alpha_1'(0)$.

  To compute $\alpha_1'(0)$ (together with $\alpha_2'(0)$), let us take the derivative of both of the equations in \eqref{eq:alphas} by their argument:
  \begin{subequations} \label{eq:alpha derivative}
  \begin{align} 
    \alpha_1' (z) &= \frac{\gamma_2'(\gamma_1(z))\gamma_1'(z)}{\alpha_2'(\alpha_2^{-1}(\gamma_2(\gamma_1(z)))} \\
    \alpha_2'(z) &= L (\gamma_1'(0) + \gamma_2'(z)) - \alpha_1'(0)
  \end{align}
  \end{subequations}
  Let us define $G = \gamma_1'(0) + \gamma_2'(0)$. We then obtain:
  \begin{equation}\label{eq:alpha discriminant}
    \alpha_{1,2}'(0) = \frac{LG \mp \sqrt{L^2 G^2 - 4 L \gamma_1'(0) \gamma_2'(0)} }{2}
  \end{equation}
\end{proof}

\begin{lemma}\label{lemma:alpha12 classK}
 If $L\geq K$ and $\gamma_1'(0)\geq \gamma_2'(0)$, then the two functions $\alpha_1$, $\alpha_2$ from \Cref{lemma:alpha12} are class $\cK$.
\end{lemma}
\begin{proof}{}
  Using the condition $L \geq K$ in \eqref{eq:alpha discriminant} implies the realness of $\alpha_1'(0), \alpha_2'(0)$.

  We also need to ask both $\alpha_1'(z)$ and $\alpha_2'(z)$ to be positive for each $z \in \real{}$. First of all, by looking at equations \eqref{eq:alpha derivative} we can notice that $\alpha_1'(z)$ has the same sign as $\alpha_2'(z)$, so we just need to ask for $\alpha_2'(z) > 0$ for each $z$. This is true if $\gamma'_2(z) > \frac{\alpha_1'(0)}{L} - \gamma_1(0)$, so a sufficient condition to ask is $\frac{\alpha_1'(0)}{L} - \gamma_1(0) < 0$. Now, taking the convention that $\alpha'_1(0)$ is obtained using the solution with the minus sign and expanding equation \eqref{eq:alpha discriminant}, the last condition is true when
  \begin{equation}
    \gamma_1'(0) \geq \gamma_2'(0) - \sqrt{G^2 - \frac{4}{L} \gamma_1'(0) \gamma_2'(0)}
  \end{equation}
  which is trivially true if $\gamma_1'(0) \geq \gamma_2'(0)$.
\end{proof}
\begin{proof}{ of \Cref{th:split_second_order}}
As in the proof of Theorem \ref{th:split_first_order}, the sum of the first two constraints recovers a formula similar to the second-order barrier function constraint, namely:
  \begin{align*}
    \ddot{h} + \gamma_0\left ( c_2 + \frac{c_1}{\derr} \right ) \left(\gamma'_1 \dot{h} +\gamma_2 (\dot{h} + \gamma_1 (h(x)) \right) - M \geq 0.
  \end{align*}

   To make use of the comparison lemma, we need to prove that there exists two, time-varying, class $\cK$ functions $\alpha_1, \alpha_2$ such that:
  \begin{multline}\label{eq:comparison bound}
    \frac{\partial \alpha_1}{\partial t} + \alpha_1'(h) \dot{h} + \alpha_2( \dot{h} + \alpha_1(h)) \geq \\
    L \left (\gamma_1'(h) \dot{h} + \gamma_2( \dot{h} + \gamma_1(h))\right) - M
  \end{multline}

  where the explicit time dependence has been omitted for conciseness,
  and $M$ is a design parameter that will be used to deal with the unknown $\frac{\de \alpha_1}{\de t}$; in fact, considering the results of Lemmata~\ref{lemma:alpha12} and~\ref{lemma:alpha12 classK}, eq. \eqref{eq:comparison bound} holds if $L\geq K$ and if $\frac{\partial \alpha_1}{\partial t}\geq M$.

  Using \eqref{eq:scale second} to substitute $c_1$ in $L \geq K$, we get: $c_2 + \frac{\gamma_0 - c_2}{\derr} \geq K$, which is equivalent to $c_2 (1 - \derr) \geq K \derr - \gamma_0$, since $\derr > 0$. As in the case analyzed in Theorem \ref{th:split_first_order}, based on the value of $\derr$ we get
  \begin{align}
    \label{eq:bounds}
    \begin{cases}
      \gamma_0 \geq K, \quad \derr = 1 \\
      c_2 \geq \frac{K \derr - \gamma_0}{ \derr - 1}, \quad \derr > 1 \\
      c_2 \leq \frac{K \derr - \gamma_0}{ \derr - 1}, \quad \derr < 1
    \end{cases}
  \end{align}

  Now, under the assumption that $\gamma_0 > K$, one can prove the function $\frac{Kd - \gamma_0}{d - 1}$ is monotonically increasing over the intervals $]-\infty, 1[$ and $]1, +\infty[$, and that it has the shape showed in \Cref{fig:Bounds}, which means $\frac{Kd - \gamma_0}{d - 1} > a $ over the interval $]-\infty, 1[$ and $\frac{Kd - \gamma_0}{d - 1} < K $ over the interval $]1, +\infty[$. This means that, given a relative error range $[d_m, d_M]$, we can impose constraint \ref{eq:bounds} as long as $d_M > 1 > d_m > 0$.

We now need to prove that $\frac{\partial \alpha_1}{\partial t} \geq -M$ for each $t$. This is equivalent to $M \geq \max_t \left \lvert \frac{\partial \alpha_1}{\partial t} \right \rvert$, which can be imposed to be true whenever $\frac{\partial \derr}{\partial t}$ is bounded.

  Now, as in Theorem \ref{th:split_first_order}, we can invoke Theorem \ref{th:var_alpha} and the comparison lemma \cite{khalilnonlinear} to conclude the proof.
\end{proof}

\begin{figure}
  \centering
  \includegraphics[width=1\linewidth]{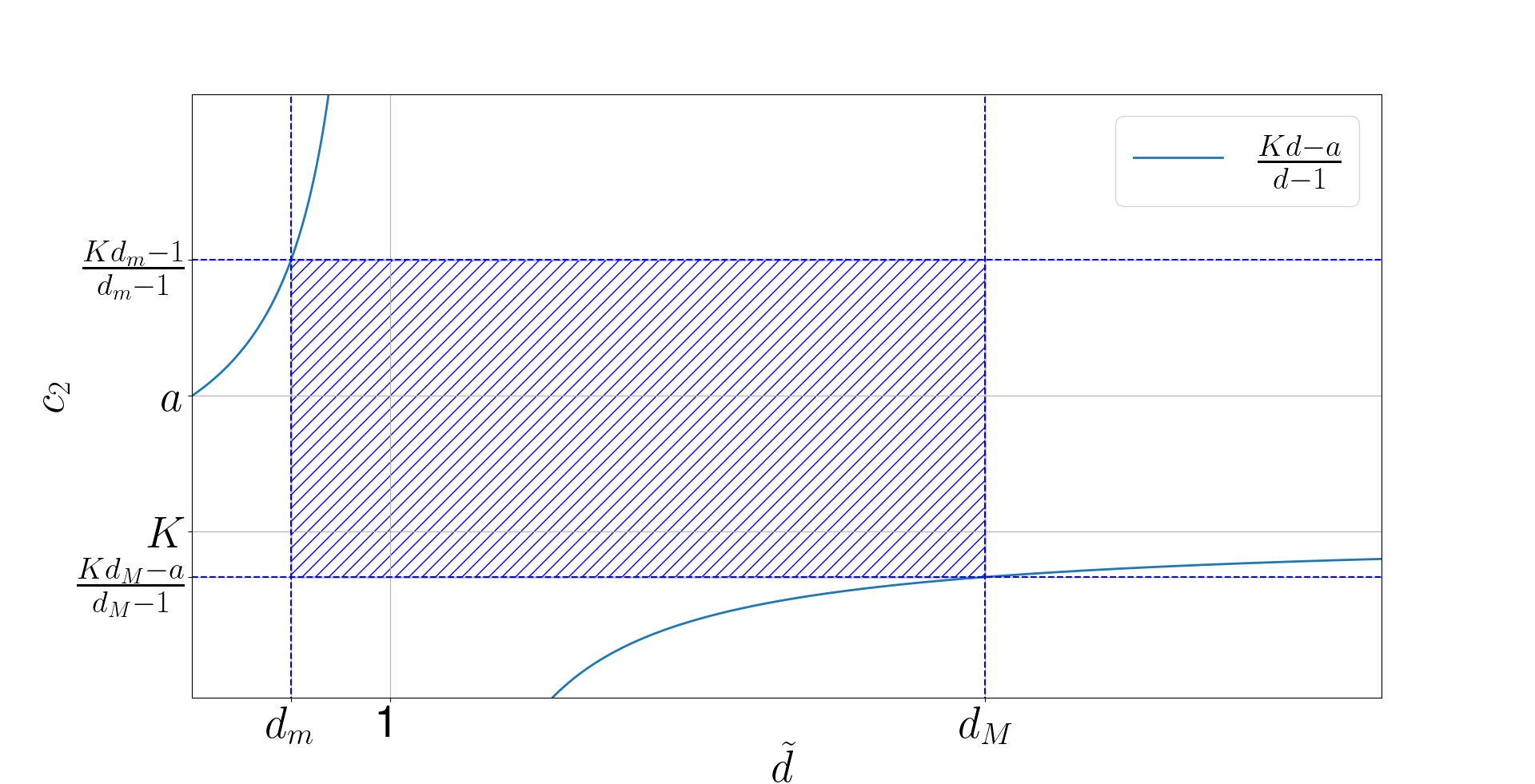}
  \caption{The figure shows the relation between the allowable error ration $\derr$ and the resulting range for $c_2$ in Theorem \ref{th:split_second_order}. As we can see, a trade-off exists between the bounds on $\derr$ and the resulting range for $c_2$.}
  \label{fig:Bounds}
\end{figure}

\subsection{Multiple Features}
When $N > 1$ features need to be tracked, we propose to extend the quadratic programming of Theorems \ref{th:split_first_order} and \ref{th:split_second_order} by adding for each feature a couple of optimization variables $c_{i1}, c_{i2}$, and the same set of constraints proposed for the case of the single features.

\begin{remark}
  Having multiple constraints may lead to the infeasibility of the optimization problem. To overcome these issues, and improve the numerical stability of the optimization, one may add some slack variables $\delta_{i1}, \delta_{i2} > 0$ to each couple of constraints to allow for constraints violation. The square of the norm of these variables should be added to the cost function, weighted by some constant chosen by the designer.
\end{remark}




\section{Numerical Simulation}
\label{sec:Experiments}
\begin{figure}
\centering
\begin{minipage}{.23\textwidth}
  \centering
  \includegraphics[width=\linewidth]{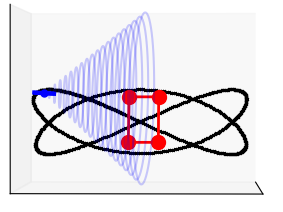}
  (a)
\end{minipage}%
\begin{minipage}{.23\textwidth}
  \centering
  \includegraphics[width=\linewidth]{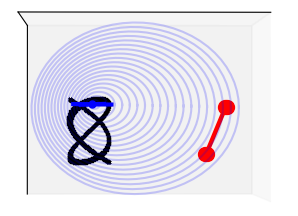}
  (b)
\end{minipage}
\begin{minipage}{.23\textwidth}
  \centering
  \includegraphics[width=\linewidth]{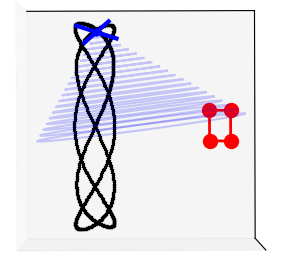}
  (c)
\end{minipage}%
\begin{minipage}{.23\textwidth}
  \centering
  \includegraphics[width=\linewidth]{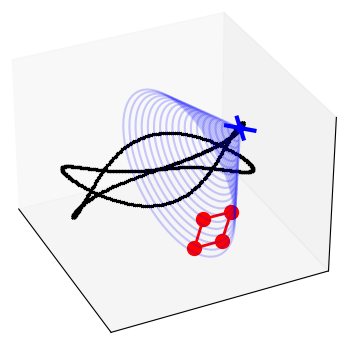}
  (d)
\end{minipage}
\caption{Snapshot of the simulation scenario of the numerical experiments. We can see that the agent (represented by a blue cross) needs to turn to keep all the features (red points) in its field of view (blue cone) to follow the reference trajectory (black line). (a) Front view. (b) Side view (c) Top-down view. (d) Tilted view.}
\label{fig:trajectories}
\end{figure}

To validate the proposed approach, in this section, we present two numerical simulations\footnote{The code used to implement the simulation can be found at https://github.com/biagio-trimarchi/A-Control-Barrier-Function-Candidate-for-Limited-Field-of-View-Sensors.git} involving the control of a double integrator system, as presented in Section \ref{sec:Approach}, and the control of a quadrotor. In both simulations, the task of each agent is to follow the same trajectory $p_{ref}(t) = {[\cos(0.3t), 10 \cos(0.2t), 2 \cos(0.2t)]}^T$ (depicted in Fig. \ref{fig:trajectories}). The features to keep in the field of view, $q_1 = {[7.0, -1.5, 1.5]}^T$, $q_2 = {[7.0, 1.5, 1.5]}^T$, $q_3 = {[6.0, 1.5, -1.5]}^T$, and $q_4 = {[6.0, -1.5, -1.5]}^T$ form the corners of a rectangle. 
We use a conic field of view with half-aperture $\psi_F = \frac{\pi}{6}$. The two simulations show an extreme case where the agents cannot estimate the distance from the features (i.e., $\hat{d} = 1$). The parameters mentioned in Theorem \ref{th:split_second_order} are $\gamma_0 = 3$, $d_M = \unit[15]{m}$, $d_m = 0.5m$ and $M = 0$.

Despite the fact that, in this extreme case, the assumptions of \Cref{th:split_second_order} are not satisfied, our approach can maintain the field of view constraints.

Fig. \ref{fig:acceleration_control}(a) and \ref{fig:acceleration_control}(b) show, respectively, the tracking error of the proposed control scheme and the minimum of the Control Barrier Functions $h_i(x)$ values in the double integrator scenario. The nominal control law of the agent is $u^*(x) = k_p (p_{ref} - p) + k_v (\dot{p}_{ref} - v) + \ddot{p}_{ref}$, with $k_p = 20.8$ and $k_v = 13.3$.

Fig. \ref{fig:quadrotor}(a) and \ref{fig:quadrotor}(b) show, respectively, the tracking error of the proposed control scheme and the minimum of the Control Barrier Functions $h_i(x)$ values in the quadrotor scenario. The nominal control law of the agent is the one proposed in \cite{Lee2010geometric}, with $k_p = 20.8$, $k_v = 13.3$., $k_R = 54.81$, and $k_\omega = 10.54$.

We can see, in both cases, that the constraint of imposed by the Control Barrier Function is always satisfied, i.e. the features never leave the field of view of the agent, and the reference trajectory is tracked with a bounded tracking error.

\begin{figure}
\centering
\begin{minipage}{.24\textwidth}
  \centering
  \includegraphics[width=\linewidth]{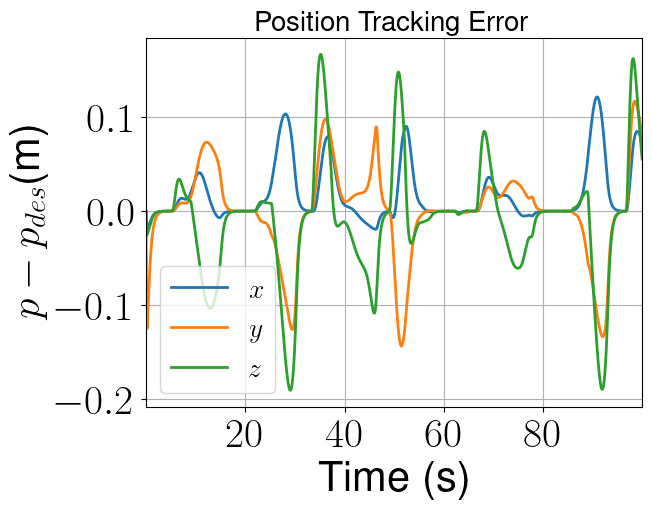}
  (a)
\end{minipage}%
\begin{minipage}{.24\textwidth}
  \centering
  \includegraphics[width=\linewidth]{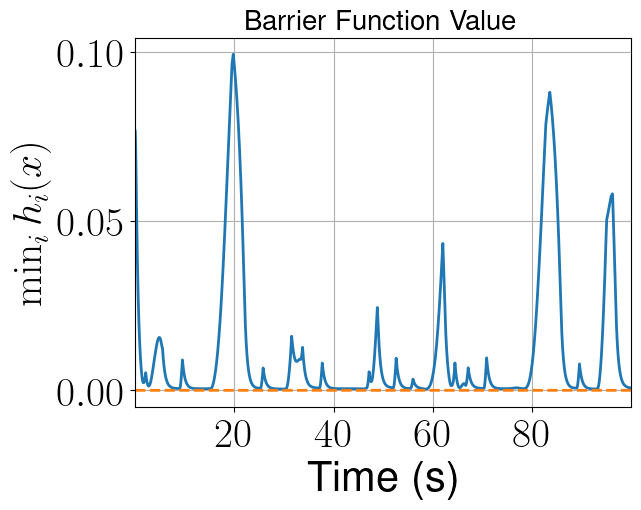}
  (b)
\end{minipage}
\caption{\emph{Acceleration Control:} Plot (a) shows the position tracking error of a rigid body actuated both in linear and angular acceleration along the prescribed path. Plot (b) shows the minimum value among the barrier functions associated with the features.}
\label{fig:acceleration_control}
\end{figure}

\begin{figure}
\centering
\begin{minipage}{.24\textwidth}
  \centering
  \includegraphics[width=\linewidth]{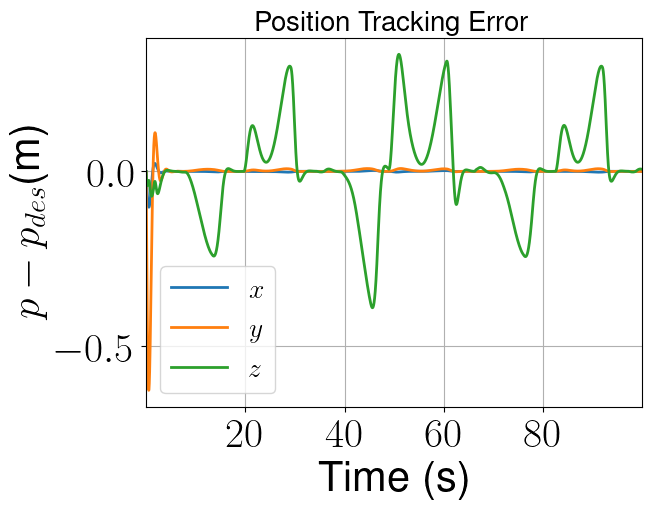}
  (a)
\end{minipage}%
\begin{minipage}{.24\textwidth}
  \centering
  \includegraphics[width=\linewidth]{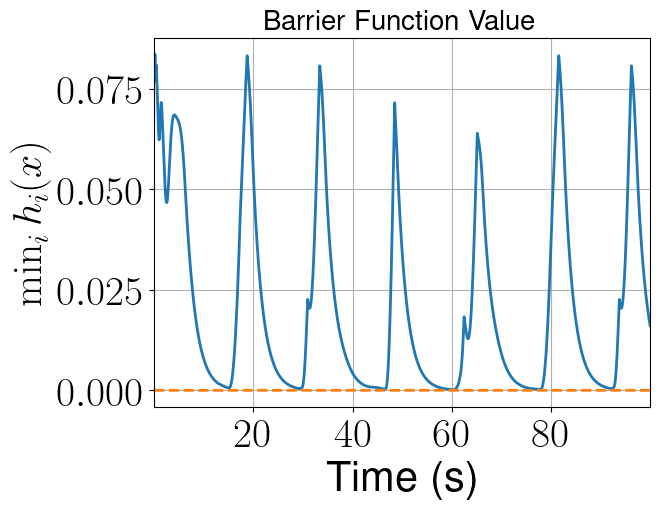}
  (b)
\end{minipage}
\caption{\emph{Quadrotor:} Plot (a) shows the position tracking error of a quadrotor actuated both in trust and torques along the prescribed path. Plot (b) shows the minimum value among the barrier functions associated with the features.}
\label{fig:quadrotor}
\end{figure}

\section{Conclusions}
\label{sec:Conclusions}
This paper presented a Control Barrier Function controller that deals with visual sensors' limited field of view. The barrier function was implemented by splitting the usual constraint into two different constraints, which allowed us to remove the effects of bounded distance measurement errors on the control law.

We supported the theoretical results by simulating the control of a fully actuated rigid body and of a quadrotor, both tasked with tracking a gate-like structure while following a trajectory. Both simulations confirmed the presented theoretical results and showed that the bounds used in the theorems may be overly conservative and could be relaxed.

Future works will focus on implementing real-world experiments to give more support to the solution presented in this work and to expand the framework to track also moving visual features. Moreover, we plan to investigate the possibility of designing a vision-based trajectory-tracking control law by transferring the theory presented in this article to the Control Lyapunov Function \cite{bahreinian2021robust} framework.




\bibliographystyle{unsrt}

\bibliography{biblio/IEEEfull,biblio/IEEEConfFull,biblio/OtherFull,
  biblio/tron,%
  biblio/formationControl,%
  biblio/websites,%
  paper_bibliography%
  }
\setcounter{section}{0}
\renewcommand\thesection{\Alph{section}}
\section{Appendix}
This section contains the computations of the terms $\grad_p h$, $\grad_R h$, $\hess_p h$, and $\hess_R h$. To understand the following equations, it is useful to keep the this equalities in mind
\begin{align*}
    \dot{\bearing} &= - \frac{\projectionM{\bearing}}{d} v \\
    \dot{d} &= - \bearing^T v
\end{align*}

According to the usual Taylor expansion definition, to find the expression of $\grad_p h$ and $\grad_R h$ we need to compute the time derivative of $h$
\begin{align*}
    \dot{h} &= -\frac{\cameraAxis^T R^T \projectionM{\beta}}{d} v + \bearing^T R \hatMap{\omega} \cameraAxis = \\
    &= -\frac{\cameraAxis^T R^T \projectionM{\beta}}{d} v - \bearing^T R \hatMap{\cameraAxis}\omega
\end{align*}

where the last equality comes from the properties of the hat map. By confronting the above terms with the Taylor expansion of $h$ and the definitions of $\grad_p h$ and $\grad_R h$, one can readily see that
\begin{align*}
    \grad_p h &= -\frac{\cameraAxis^T R^T \projectionM{\beta}}{d} \\
    \grad_R h &= -\bearing^T R \hatMap{\cameraAxis}
\end{align*}

The formula for the hessian can be found by computing the second derivative of $h$

\begin{align*}
    \ddot{h} &= -\frac{\cameraAxis^T R^T \projectionM{\beta}}{d} a + \frac{d}{dt} \left ( -\frac{\cameraAxis^T R^T \projectionM{\beta}}{d} \right ) v + \\ 
    & - \bearing^T R \hatMap{\cameraAxis} \alpha -\bearing^T R \hatMap{\omega} \hatMap{\cameraAxis} \omega - \dot{\bearing}^T R \hatMap{\cameraAxis} \omega.
\end{align*}
To find the expressions of the hessian of $h$, we need to expand and rearrange some of the above terms.

First of all, using again the properties of the hat map
\begin{align*}
    -\bearing^T R \hatMap{\omega} \hatMap{\cameraAxis} \omega = \omega^T \hatMap{R \bearing} \hatMap{\cameraAxis} \omega
\end{align*}
we can already derive $\metrica{\hess_R h [\omega]}{\omega}= \omega^T \hatMap{R^T \bearing} \hatMap{\cameraAxis} \omega$.

Before going on, let us define $z = R \cameraAxis$, allowing us to write the following expressions more compactly. Moreover, notice that $\projectionM{\bearing}z = z - (\bearing^T z) \bearing$. 

Now, we can expand the term with the time derivative
\begin{align*}
    \frac{d}{dt} \left ( -\frac{\cameraAxis^T R^T \projectionM{\beta}}{d} \right ) v = -v^T \frac{\left( \projectionM{\bearing} z \right) \bearing^T}{d^2} v - \frac{\dot{z}^T \projectionM{\bearing}}{d} v - \frac{z^T \dot{\projectionM{\bearing}}}{d} v
\end{align*}
and by noticing that
\begin{align*}
    -\frac{\dot{z}^T \projectionM{\bearing}}{d} v = -\dot{\bearing}^T R \hatMap{\cameraAxis} \omega = v^T \frac{\projectionM{\bearing} R \hatMap{\cameraAxis}}{d} \omega
\end{align*}
we can write $\metrica{\hess_R h[\omega]}{v} = \metrica{\hess_p h[v]}{\omega}= v^T \frac{\projectionM{\bearing} R \hatMap{\cameraAxis}}{d} \omega$. To find the last component of the hessian, we need to expand the remaining two terms
\begin{align*}
    -v^T\frac{\left( \projectionM{\bearing} z \right) \bearing^T}{d^2} v &= -v^T\frac{\left( z - (\bearing^T z) \bearing \right) \bearing^T}{d^2} v = \\
    = -v^T&\frac{z \bearing^T}{d^2} v + v^T\frac{\left( (\bearing^T z) \bearing \right) \bearing^T}{d^2} v = \\
    = -v^T&\frac{z \bearing^T}{d^2} v - v^T\frac{(\bearing^T z) \projectionM{\bearing}}{d^2}  v + v^T\frac{(\bearing^T z) I}{d^2}  v \\
    - \frac{z^T \dot{\projectionM{\bearing}}}{d} v &= \frac{(z^T\dot{\bearing}) \bearing^T}{d} v + \frac{(z^T\bearing) \dot{\bearing}^T}{d} v = \\
    = - v^T&\frac{\bearing(z^T \projectionM{\bearing})}{d^2} v - v^T\frac{(z^T\bearing) \projectionM{\bearing}}{d^2} v = \\
    = - v^T&\frac{\bearing( z^T - (\bearing^T z) \bearing^T )}{d^2} v - v^T\frac{(z^T\bearing) \projectionM{\bearing}}{d^2} v = \\
    = - v^T&\frac{\bearing z^T }{d^2} v + v^T \frac{(\bearing^T z)I )}{d^2} v -v^T\frac{2(z^T\bearing) \projectionM{\bearing}}{d^2} v
\end{align*}

By combining the terms, we find the final expression for $\metrica{\hess_p h [v]}{v}$, which is the same reported in \eqref{eq:hessian_velocity}.


\end{document}